\documentclass[preprint,12pt]{elsarticle}
\usepackage[T1]{fontenc}
\usepackage[english]{babel}
\usepackage{ae,aecompl}
\usepackage{amsmath, amsthm}
\usepackage{amssymb}
\usepackage[utf8]{inputenc}
\usepackage[section] {placeins}
\usepackage{hyperref}
\hypersetup{colorlinks=true} 
\usepackage[ruled, vlined, linesnumbered]{algorithm2e}
\newtheorem {Theorem}                 {Theorem}         [section]
\newtheorem {theorem}      [Theorem]  {Theorem}
\newtheorem {meinAlgorithmus}    [Theorem]  {Algorithm}

\newtheorem {Corollary}    [Theorem]  {Corollary}

\newtheorem {Problem}        [Theorem]  {Problem}

\newtheorem {Fact}         [Theorem]  {Fact}

\newtheorem {Lemma}        [Theorem]  {Lemma}

\input amssym.def
\input amssym

\renewcommand\qed{\hspace*{\fill}$\Box$}

\usepackage{pgfplots}
\usepackage{verbatim}
\usepackage{tikz}
\usetikzlibrary{shapes,arrows,backgrounds,mindmap}
\usepackage{titlesec}
\journal{arXiv}

\begin{document}
\begin{frontmatter}
\title{On computing the $2$-vertex-connected components of directed graphs}
\author{Raed Jaberi\corref{cor1}}
\cortext[cor1]{Tel.: +49 3677 69 -2786; fax: +49 3677 69 -1237.}
\address{Faculty of Computer Science and Automation, Teschnische Universität Ilmenau, 
\\$98694$ Ilmenau, Germany}
\ead{raed.jaberi@tu-ilmenau.de}
\begin{abstract} 
In this paper we consider the problem of computing the $2$-vertex-connected components ($2$-vccs) of directed graphs. We present two new algorithms for solving this problem. The first algorithm runs in $O(mn^{2})$ time, the second in $O(nm)$ time. Furthermore, we show that the old algorithm of Erusalimskii and Svetlov runs in $O(nm^{2})$ time. In this paper, we investigate the relationship between $2$-vccs and dominator trees. We also present an algorithm for computing the $3$-vertex-connected components ($3$-vccs) of a directed graph in $O(n^{3}m)$ time, and we show that the $k$-vertex-connected components ($k$-vccs) of a directed graph can be computed in $O(mn^{2k-3})$ time. Finally, we consider three applications of our new algorithms, which are approximation algorithms for problems that are generalization of the problem of approximating the smallest $2$-vertex-connected spanning subgraph of $2$-vertex-connected directed graph.
\end{abstract}
\end{frontmatter}
\section{Introduction}
Let $G=(V,E)$ be a directed graph with $|V|=n$ vertices and $|E|=m$ edges. A strong articulation point of $G$ is a vertex whose removal increases the number of strongly connected components of $G$. A directed graph $G=(V,E)$ is said to be $k$-vertex-connected if it has at least $k+1$ vertices and the induced subgraph on $V\setminus X$ is strongly connected for every $X \subset V $ with $|X|<k$. Thus, a strongly connected graph $G=(V,E)$ is $2$-vertex-connected if and only if it has at least $3$ vertices and it contains no strong articulation points. The $2$-vertex-connected components of a strongly connected graph $G$ are its maximal $2$-vertex-connected subgraphs. The concept was defined in \cite{lite6}. For more information see \cite{lite2}.

In $2010$, Georgiadis \cite{lite3} gave a linear time algorithm to test whether a strongly connected graph $G$ is $2$-vertex-connected or not. Later, Italiano et al. \cite{lite2} gave a linear time algorithm for the same problem which is faster in practice than the algorithm of Georgiadis \cite{lite3}. Moreover, the algorithm of Italiano et al. \cite{lite2} can find all the strong articulation points of a directed graph $G$ in linear time. The algorithm from \cite{lite2} solved an open problem posed by Beldiceanu et al. ($2005$) \cite{lite4}.
In $1980$, Erusalimskii and Svetlov \cite{lite6} gave an algorithm for computing the $2$-vccs of a directed graph, whose running time was not analyzed. In this work we show that this algorithm runs in $O(nm^{2})$ time. Furthermore, we present two new algorithms for computing the $2$-vccs of a directed graph. The first algorithm runs in $O(mn^{2})$ time, and the second in $O(nm)$ time. The question posed by Italiano et al. \cite{lite2} whether the problem is solvable in linear time still remains open.

The remainder of this paper is organized as follows. In section \ref{def:Psec2}, we briefly describe the algorithm of Italiano et al. \cite{lite2} for finding the strong articulation points of a directed graph.
In section \ref{def:Psec3}, we briefly describe the algorithm of Erusalimskii and Svetlov \cite{lite6} for
computing the $2$-vccs of a directed graph and analyze its running time.
In section \ref{def:Psec4}, we present a new algorithm for computing the $2$-vccs that contain a certain vertex if such a component exists. Then we use this algorithm to compute all the $2$-vccs of a directed graph in $O(mn^{2})$ time. In section \ref{def:Psec5}, we present another new algorithm for computing all the $2$-vccs of a directed graph in $O(nm)$ time. Afterwards, we investigate the relationship between $2$-vccs and dominator trees in section \ref{def:rel2vccsdts}.
In section \ref{def:Psec7}, we present an algorithm for computing the $3$-vccs of a directed graph in $O(n^{3}m)$ time, and we show that the $k$-vccs of a directed graph can be computed in $O(mn^{2k-3})$ time. Finally in section \ref{def:Psec8}, we consider three applications of our new algorithms, which are approximation algorithms for problems that are generalization of the problem of approximating the smallest $2$-vertex-connected spanning subgraph of $2$-vertex-connected directed graph.
\section{The algorithm of Italiano et al.} \label{def:Psec2}
In this section, we briefly describe the algorithm of Italiano et al. \cite{lite2} for computing all the strong articulation points of a directed graph. This algorithm will be used later. The content of this section is based on  \cite{lite2}. We first explain some definitions and notations which will be used throughout this paper.
A \textit{flowgraph} $G(v)=(V,E,v)$ is a directed graph with $|V|=n$ vertices, $|E|=m$ edges, and a distinguished start vertex $v \in V$ such that every vertex $w\in V$ is reachable from $v$. For a flowgraph $G(v)=(V,E,v)$, the \textit{dominance relation} of $G(v)$ is defined as follows: a vertex $w \in V$ is a \textit{dominator} of vertex $u \in V$ if every path from $v$ to $u$ includes $w$. By $dom(w)$ we denote the set of dominators of vertex $w$. Obviously, the set of dominators of the start vertex in $G(v)$ is $dom(v)=\lbrace v \rbrace$.
For every vertex $w\in V$ with $w \neq v$, $\lbrace v,w \rbrace$ is a subset of $dom(w)$; we call $w,v$ the \textit{trivial dominators} of $w$. A vertex $u$ is a \textit{non-trivial dominator} in $G(v)$ if there is some $w \notin \lbrace v,u\rbrace$ such that $u \in dom(w)$. The set of all non-trivial dominators is called $D(v)$. The dominance relation is reflexive, transitive, and antisymmetric. A vertex $u\in V$ is an \textit{immediate dominator} of vertex $w \in V$ in $G(v)$ if $u\in dom(w)$ and all elements of $dom(w)\setminus \lbrace w \rbrace$ are dominators of $u$. Every vertex $w$ of $G(v)$ except the start vertex $v$ has a unique immediate dominator, which is denoted by $imd(w)$. The edges $(u,w)$, where $u$ is the immediate dominator of $w$, form a tree with root $v$, called the \textit{dominator tree} of $G(v)$, denoted by $DT(v)$. Vertex $w \in V$ is a dominator of vertex $u \in V$ in $G(v)$ if and only if $w$ is an ancestor of $u$ in $DT(v)$.

Let $G=(V,E)$ be a directed graph. Let $F$ be a subset of $E$ and let $U$ be a subset of $V$. We use 
$G\setminus F$ to denote the directed graph obtained from $G$ by deleting all the edges in $F$. We use $G\setminus U$ to denote the directed graph obtained form $G$ by removing all the vertices in $U$ and their incident edges. By $G[F]$ we denote the directed graph $(V[F],F)$ whose $V[F]=\lbrace w\mid \exists u\in V:(w,u) \in F $ or $(u,w) \in F \rbrace$. By $G[U]$ we denote the directed graph $(U,E[U])$ whose $E[U]=\lbrace (w,u) \mid w,u \in U $ and $ (w,u) \in E \rbrace$. $G[F]$ and $G[U]$ are subgraphs of $G$. The reversal graph of $G$ is the directed graph 
$G^{R}=(V,E^{R})$, where $E^{R}=\lbrace (w,u) \mid (u,w) \in E \rbrace$.

Let $G=(V,E)$ be a strongly connected graph and let $v$ be a vertex in $G$. Since $G^{R}$ is strongly connected, $G^{R}(v)=(V,E^{R},v)$ is a flowgraph. By $D^{R}(v)$ we denote the set of all non-trivial dominators in the flowgraph $G^{R}(v)$. \\
The algorithm of Italiano et al. is based on the following fact.
\begin{Fact}\label{def:satzvonitaliano}
\begin{rm}\cite{lite2}
\end{rm}Let $G=(V,E)$ be a strongly connected graph, and let $v$ be any vertex in $G$. Then vertex $w\in V$
with $w \neq v$ is a strong articulation point if and only if $w$ is a non-trivial dominator in the flowgraph $G(v)$ or in the flowgraph $G^{R}(v)$.
\end{Fact} 
The set of the strong articulation points of an arbitrary directed graph $G$ is the union of the strong articulation points of its
strongly connected components.
Algorithm \ref{algo:SAKItaliano} shows the algorithm of Italiano et al. \cite{lite2}. More information on this algorithm can be found in \cite{lite2,lite13}.
\begin{figure}[htbp]
\begin{meinAlgorithmus}[\textbf{SAVs$(G)$}]\begin{rm}(from \cite{lite2})\end{rm}\label{algo:SAKItaliano}\rm\quad\\[-5ex]
\begin{tabbing}
\quad\quad\=\quad\=\quad\=\quad\=\quad\=\quad\=\quad\=\quad\=\kill
\textbf{Input:} A strongly connected graph $G=(V,E)$.\\
\textbf{Output:} The strong articulation points of $G$.\\
{\small 1}\> Choose $v \in C$ arbitrarily.\\
{\small 2}\> \textbf{if} $G\setminus \lbrace v \rbrace$ is not strongly connected \textbf{then} \textbf{output} $v$.\\
{\small 3}\> Compute and \textbf{output} $D(v)$.\\
{\small 4}\> Calculate the reversal graph $G^{R}$.\\
{\small 5}\> Compute and \textbf{output} $D^{R}(v)$.
\end{tabbing}
\end{meinAlgorithmus}
\end{figure}   
\begin{Fact}\label{def:AlgoItalianoLaufzeit}
\begin{rm}\cite{lite2}
\end{rm}Algorithm \ref{algo:SAKItaliano} runs in $O(n+m)$ time.
\end{Fact} 
\section{Algorithm of Erusalimskii and Svetlov} \label{def:Psec3}
In this section, we briefly describe the algorithm of Erusalimskii and Svetlov \cite{lite6} for computing the $2$-vccs of a directed graph, and we analyze its running time. The latter analysis was missing in \cite{lite6}. The content of this section is based mainly on \cite{lite6}.
In \cite{lite6}, the authors provided an algorithm for computing all \textit{biblocks} of a directed graph, where the biblocks of a directed graph are its maximal strongly connected subgraphs that do not contain any strong articulation point. A biblock is either a $2$-vcc, a single vertex or two vertices which are connected by two antiparallel edges. 
In this paper we are only interested in computing the $2$-vccs of a directed graph.
Let $G=(V,E)$ be a strongly connected graph, and let $v$ be a strong articulation point in $G$. Then the  vertex $v$ does not necessarily occur in two or more $2$-vccs of $G$ \cite{lite6}. Moreover, $2$-vccs have the following property:
\begin{Fact}\label{def:2VCsAtMostOneCommonVertex}
\begin{rm}\cite{lite6} \end{rm}Let $C_1^{2vc},C_2^{2vc}$ be distinct $2$-vccs in directed graph $G=(V,E)$. Then $C_1^{2vc}$ and $C_2^{2vc}$ have at most one vertex in common. 
\end{Fact} 
In \cite{lite6}, the authors studied a class of directed graphs $L$ defined as follows:
A directed graph $G=(V,E)$ belongs to class $L$ if it satisfies the following conditions:
\begin{enumerate}[1.]
\item If $C_1,C_2,\ldots,C_t$ are the strongly connected components of $G$, then there are no edges between $C_i$ and $C_j$ for distinct $i,j\in \lbrace 1,\ldots,t\rbrace$.
\item For every strong articulation point $v$ the directed graph $G\setminus \lbrace v \rbrace$ satisfies $(1)$. 
\end{enumerate}
Let $G=(V,E)$ be a directed graph. By $U(G)$ we denote the undirected graph formed from $G$ by deleting 
the directions of the edges. In \cite{lite6}, the following was proved:
\begin{Fact}\label{def:2VC2CUG}
\begin{rm}\cite{lite6} \end{rm}Let $G=(V,E) \in L$ be a directed graph. The vertices of the $2$-vccs of $G$ coincide with the vertices of the $2$-connected components (i.e. $2$-vertex-connected components \cite{lite11}) of the undirected version $U(G)$.
\end{Fact} 
The main idea behind the algorithm of Erusalimskii and Svetlov \cite{lite6} is as follows. Given a directed graph $G=(V,E)$, the algorithm computes a directed graph $G' \in L$ such that the $2$-vccs of $G$ coincide with the $2$-vccs of $G'$. Then all $2$-vccs of $G'$ can be easily computed by using Fact \ref{def:2VC2CUG}.  Algorithm \ref{algo:AlgoErusalimskiiSvetlov} shows this algorithm \cite{lite6}:
\begin{figure}[htbp]
\begin{meinAlgorithmus}[\textbf{ErusalimskiiSvetlov$(G)$}]\label{algo:AlgoErusalimskiiSvetlov} \begin{rm}(from \cite{lite6}) \end{rm}\rm\quad\\[-5ex]
\begin{tabbing}
\quad\quad\=\quad\=\quad\=\quad\=\quad\=\quad\=\quad\=\quad\=\quad\=\kill
\textbf{Input:} A directed graph $G=(V,E)$.\\
\textbf{Output:} The $2$-vccs of $G$.\\
{\small 1}\> \textbf{Repeat}\\
{\small 2}\>\> Compute the strongly connected components of $G$. \\
{\small 3}\>\> Remove from $G$ the edges between the strongly connected components of $G$. \\
{\small 4}\>\> \textbf{for} every vertex $v \in V$ \textbf{do}\\
{\small 5}\>\>\>\> Compute the strongly connected components of $G \setminus \lbrace v\rbrace$.\\
{\small 6}\>\>\>\> Remove from $G$ the edges between the strongly connected  \\
{\small 7}\>\>\>\> components  of $G \setminus \lbrace v\rbrace$.\\
{\small 8}\> \textbf{until} no edge was removed during step $6$.\\
{\small 9}\> We obtain a directed graph $G'\in L$.\\
{\small 10}\> Compute the $2$-connected components $ C^{2vc}_1,C^{2vc}_2,\ldots,C^{2vc}_k $ of $U(G')$.\\
{\small 12}\> \textbf{Output} $ C^{2vc}_1,C^{2vc}_2,\ldots,C^{2vc}_k $.
\end{tabbing}
\end{meinAlgorithmus}
\end{figure} 
\begin{Fact}\label{def:korrektheitEundS}
\begin{rm}\cite{lite6} \end{rm}Let $G=(V,E)$ be a directed graph and let $G'\in L$ be the directed graph obtained after running algorithm \ref{algo:AlgoErusalimskiiSvetlov} on the input graph $G$ (in step $9$), then the $2$-vccs of $G$ coincide with the $2$-vccs of $G'$.
\end{Fact}
\begin{Theorem}\label{def:AlgoRaedJaberiLaufzeit}
The running time of algorithm \ref{algo:AlgoErusalimskiiSvetlov} is $O(nm^{2})$.
\end{Theorem} 
\begin{proof} The number of iterations of the repeat-loop is at most $m$ since at least one edge is removed in each iteration. The strongly connected components of a directed graph can be found in linear time using Tarjan's algorithm \cite{lite5}. In each iteration of the repeat-loop, steps $4$--$7$ require $O(n(n+m))$ time. The $2$-connected components of an undirected graph can be computed in linear time using Tarjan's algorithm \cite{lite5}. Thus, the total running time of algorithm \ref{algo:AlgoErusalimskiiSvetlov} is $O(m(n(m+n)))=O(nm^{2})$. 
\end{proof}
\section{Computing $2$-vertex-connected components that contain a certain vertex} \label{def:Psec4}
In this section, we present a new algorithm for computing all the $2$-vccs of a directed graph $G=(V,E)$ that contain a certain vertex $v \in V$. Note that it can happen that a vertex is not contained in any $2$-vcc, as Figure \ref{fig:KnoteninkeinerK} illustrates.
   \begin{figure}[htbp]
    \centering
    \begin{tikzpicture}[xscale=2,yscale=1]

        \tikzstyle{every node}=[color=black,draw,circle,minimum size=0.65cm]
        \node (v0) at (0,0) {$0$};
        \node (v1) at (2,-1) {$1$};
        \node (v2) at (2.5,1.5) {$2$};
        \node (v3) at (1, 2) {$3$};
        \node (v4) at (1,-1.5) {$4$};
        \node (v5) at (2,0.5) {$5$};
        \node (v6) at (-1,1.5) {$6$};
        \node (v7) at (-1,-1.5) {$7$};
       \begin{scope}
       
            \tikzstyle{every node}=[auto=right]   
            \draw [-triangle 45] (v4) to  (v7);       
             \draw [-triangle 45] (v0) to  (v3); 
             \draw [-triangle 45] (v3) to  (v5);
             \draw [-triangle 45] (v4) to (v0);
             \draw [-triangle 45] (v4) to (v5);
             \draw [-triangle 45] (v5) to  [bend right ](v4);
             \draw [-triangle 45] (v6) to (v0);
             \draw [-triangle 45] (v0) to  [bend right ](v6);
             \draw [-triangle 45] (v7) to (v0);
             \draw [-triangle 45] (v0) to  [bend right ](v7);
             \draw [-triangle 45] (v7) to (v6);
             \draw [-triangle 45] (v6) to  [bend right ](v7);
             \draw [-triangle 45] (v3) to  [bend right ](v0);             
             \draw [-triangle 45] (v5) to  [bend right ](v3);
             \draw [-triangle 45] (v0) to  [bend right ](v4);
             \draw [-triangle 45] (v2) to  [bend right ](v3);
             \draw [-triangle 45] (v4) to  [bend right ](v1);
             \draw [-triangle 45] (v1) to [bend right ]  (v2);

        \end{scope}
    \end{tikzpicture}
     \caption{Vertices $1,2$ do not lie in any $2$-vcc.}
     \label{fig:KnoteninkeinerK}
   \end{figure}
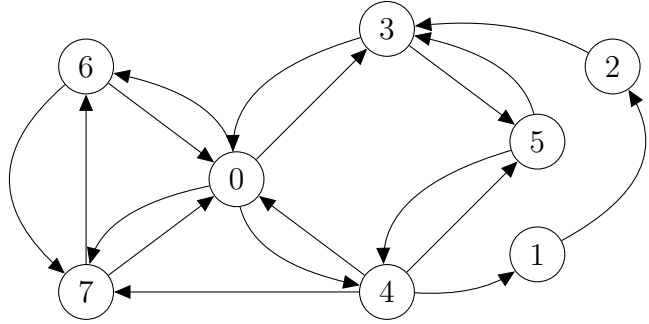
We may always assume that $G$ is strongly connected.\\
Let $G=(V,E)$ be a strongly connected graph and let $v$ be an arbitrary vertex in $G$. We consider the dominator tree $DT(v)$ of the flowgraph $G(v)=(V,E,v)$. By $K(v)$ we denote the set of direct successors of the root $v$ in the dominator tree $DT(v)$. A vertex $w\in V$ belongs to the set $K(v)$ if and only if $(v,w)\in E$ or there exist two vertex-disjoint paths form $v$ to $w$ in $G$. Let $G^{R}=(V,E^{R})$ be the reversal graph of $G$. We consider the dominator tree $DT^{R}(v)$ of $G^{R}(v)$ and denote by $K^{R}(v)$ the set of direct successors of the root $v$ in $DT^{R}(v)$ .
\begin{Lemma}\label{def:VKv2vc}
Let $G=(V,E)$ be a strongly connected graph and let $v$ be an arbitrary vertex in $G$. Then only elements of  $((K(v) \cap K^{R}(v)) \cup \lbrace v \rbrace)$ can belong to the $2$-vccs which contain the vertex $v$.
\end{Lemma} 
\begin{proof} Let $w \in V$, and assume that $w\notin  (K(v) \cap K^{R}(v)) \cup \lbrace v \rbrace$, i.e. $w\notin K(v) \cap K^{R}(v)$ and $w\neq v$. Then $w \in (V\setminus (K(v)\cup \lbrace v \rbrace)\cup (V\setminus (K^{R}(v)\cup \lbrace v \rbrace)$. There are two cases to consider:
\begin{enumerate}[1.]
\item $w\notin  K(v)\cup \lbrace v \rbrace$. Then there exists a non-trivial dominator $s$ such that every path from $v$ to $w$ includes $s$ in $G$. Therefore, there are no two vertex-disjoint paths from $v$ to $w$ in $G$. 
\item $w\notin K^{R}(v)\cup \lbrace v \rbrace$. We argue as in case $(1)$.
\end{enumerate}
\end{proof}
\begin{Lemma}\label{def:RSAKszken2Vfkenv}
Let $G=(V,E)$ be a directed graph and let $v$ be a strong articulation point in $G$. Let $C^{2vc}$ be a $2$-vcc of $G$ with $v\in C^{2vc}$. Then all vertices of $C^{2vc}\setminus \lbrace v\rbrace$ lie in a strong connected component $C$ of $G \setminus \lbrace v\rbrace$, i.e. $C^{2vc}\setminus \lbrace v\rbrace \subseteq C$. 
\end{Lemma}
\begin{proof} Since $C^{2vc}$ is a $2$-vcc of $G$, the directed graph $G[C^{2vc}]$ does not contain any articulation point. Thus, $G[C^{2vc}\setminus \lbrace v\rbrace]$ is strongly connected. Moreover, $G[C^{2vc}\setminus \lbrace v\rbrace]$ is a subgraph of $G \setminus \lbrace v\rbrace$. Consequently, $C^{2vc}\setminus \lbrace v\rbrace$ is a subset of a strongly connected component of $G \setminus \lbrace v\rbrace$.
\end{proof}
Now we can describe our algorithm for computing the $2$-vccs of a directed graph $G=(V,E)$ that contain $v$.
 
\begin{figure}[htbp]
\begin{meinAlgorithmus}[\textbf{2VCCsAlgorithm$1(G)$}]\label{algo:Algorithmus1RaedJaberi}\rm\quad\\[-5ex]
\begin{tabbing}
\quad\quad\=\quad\=\quad\=\quad\=\quad\=\quad\=\quad\=\quad\=\quad\=\kill
\textbf{Input:} A directed graph $G=(V,E)$ and a vertex $v \in V$.\\
\textbf{Output:} The $2$-vccs that contain $v$.\\
{\small 1}\> $G=(V,E)\leftarrow $ the strongly connected component $C_v$ of $G$ with $v\in C_v$.\\
{\small 2}\> \textbf{if} $G$ is $2$-vertex-connected \textbf{then} \\
{\small 3}\>\>\> \textbf{Output} $V$. \\
{\small 4}\> \textbf{else if} $v$ is not a strong articulation point in $G$ and $|K(v) \cap K^{R}(v)|\geq 2$ \textbf{then}\\
{\small 5}\>\>\> Recursively compute the $2$-vccs of $G[(K(v) \cap K^{R}(v))\cup \lbrace v \rbrace]$ that \\ 
{\small  }\>\>\>\>\> contain $v$ and \textbf{output} them.\\ 
{\small 6}\> \textbf{else if} $v$ is a strong articulation point in $G$ \textbf{then}\\
{\small 7}\>\>\> Compute the strongly connected components of $G \setminus \lbrace v \rbrace$.\\
{\small 8}\>\>\> \textbf{for} every strongly connected component $C$ of $G \setminus \lbrace v \rbrace$ \textbf{do}\\
{\small 9}\>\>\>\>\> \textbf{if} $G[C \cup  \lbrace v \rbrace]$ is strongly connected and $|C|\geq 2$ \textbf{then}\\
{\small 10}\>\>\>\>\>\>\> Recursively compute the $2$-vccs of $G[C \cup \lbrace v \rbrace]$ that contain $v$. \\
{\small 11}\>\>\> \textbf{Output} all the computed $2$-vccs.
\end{tabbing}
\end{meinAlgorithmus}
\end{figure}
Algorithm \ref{algo:Algorithmus1RaedJaberi} works as follows. First, line $1$ finds the strongly connected component $C_v$ of $G$ with $v\in C_v$ using Tarjan's algorithm \cite{lite5} and assigns the directed graph $G[C_v]$ to $G$ because all the $2$-vertex-connected components which contain $v$ lie in $G[C_v]$. Then, line  $2$ tests whether $G$ is $2$-vertex-connected using the algorithm of Italiano et al. \cite{lite2}, and if it is, line $3$ outputs $V$. Otherwise, the algorithm tests whether $v$ is a strong articulation point in $G$ or not. If $v$ is not a strong articulation point in $G$ and $|K(v) \cap K^{R}(v)|\geq 2$, then line $5$ recursively computes the $2$-vccs of $G[(K(v) \cap K^{R}(v))\cup \lbrace v \rbrace]$ that contain $v$ and outputs them. This is correct by Lemma \ref{def:VKv2vc}. If $v$ is a strong articulation point in $G$, then the for loop of lines $8$--$10$ recursively computes the $2$-vccs of $G[C \cup \lbrace v \rbrace]$ that include $v$ for every strongly connected component $C$ of $G \setminus \lbrace v \rbrace$, where $G[C \cup  \lbrace v \rbrace]$ is strongly connected and $|C|\geq 2$. This is correct by Lemma \ref{def:RSAKszken2Vfkenv}.  
\begin{Theorem}\label{def:Laufzeitunterschied}
 Algorithm \ref{algo:Algorithmus1RaedJaberi} runs in $O(nm)$ time.
\end{Theorem} 
\begin{proof} The dominators of a flowgraph can be found in linear time \cite{lite1,lite12}. The strong articulation points of a directed graph can also be computed in linear time using the algorithm of Italiano et al. \cite{lite2}. Furthermore, the strongly connected components of a directed graph can be computed in linear time using Tarjan's algorithm \cite{lite5}. At each level of the recursion at least one vertex must be removed in lines $4$--$5$ or the set of vertices must be split in lines $6$--$10$. Hence, the recursion depth is at most $n$. Fix some recursion level. We consider the cost of the calls of the procedure excepting the recursion. For one call, the cost is linear in the size of the current subgraph. Let $G'[C_1 \cup  \lbrace v \rbrace]=(V_1,E_1),G'[C_2 \cup  \lbrace v \rbrace]=(V_2,E_2),\ldots,G'[C_t \cup  \lbrace v \rbrace]=(V_t,E_t)$ be the subgraphs of the directed graph $G'=(V',E')$ considered on this level in all calls. Then $\sum _{1 \leq i \leq t} |E_i|\leq |E'|$ since the strongly connected components of $G'$ are disjoint. The total cost at each level of the recursion is therefore $O(m)$. 
\end{proof}
\begin{Corollary}\label{def:All2VCsAlgoRJaberi}
Let $G=(V,E)$ be a directed graph. If we apply algorithm \ref{algo:Algorithmus1RaedJaberi} for every vertex $v \in V$, we can find all the $2$-vccs of $G$ in $O(n^{2}m)$ time.
\end{Corollary} 
\begin{proof} This follows immediately from Theorem \ref{def:Laufzeitunterschied}.
\end{proof}
\section{Computing $2$-vertex-connected components of directed graphs}\label{def:Psec5}
In this section, we present a new algorithm for computing all the $2$-vccs of a directed graph in $O(nm)$ time. Our algorithm is based on the following Lemma.
\begin{Lemma}\label{def:RSAKszken2Vfken}
Let $G=(V,E)$ be a directed graph and let $w$ be a strong articulation point in $G$.
Let $C^{2vc}$ be a $2$-vcc of $G$. Then all vertices of $ C^{2vc}\setminus \lbrace w\rbrace$ lie in a strongly connected component $C$ of $G \setminus \lbrace w\rbrace$, i.e. $ C^{2vc}\setminus \lbrace w\rbrace \subseteq C$.
\end{Lemma} 
\begin{proof} As in Lemma \ref{def:RSAKszken2Vfkenv}.  
\end{proof}
\begin{Corollary}\label{def:CorollaryRSAKszken2Vfken}
Let $G=(V,E)$ be a directed graph and let $w$ be a strong articulation point in $G$. 
The $2$-vccs of $G$ lie in the subgraphs $G[C_1 \cup \lbrace w\rbrace],G[C_2 \cup \lbrace w\rbrace],\ldots,G[C_t \cup \lbrace w\rbrace]$, where $C_1,C_2,\ldots,C_t$ are the strongly connected components of $G \setminus \lbrace w\rbrace$. 
\end{Corollary}
We now describe our algorithm for computing all the $2$-vccs of a directed graph $G=(V,E)$. 
\begin{figure}[htbp]
\begin{meinAlgorithmus}[\textbf{2VCCsAlgorithm2$(G)$}]\label{algo:simpleAlgorithmus2VsCCsRaedJaberi}\rm\quad\\[-5ex]
\begin{tabbing}
\quad\quad\=\quad\=\quad\=\quad\=\quad\=\quad\=\quad\=\quad\=\quad\=\kill
\textbf{Input:} A directed graph $G=(V,E)$.\\
\textbf{Output:} The $2$-vccs of $G$.\\
{\small 1}\> \textbf{if} $G$ is $2$-vertex-connected \textbf{then} \\
{\small 2}\>\> \textbf{Output} $V$. \\
{\small 3}\> \textbf{else}\\
{\small 4}\>\> Find a strong articulation point $w$ of $G$.\\
{\small 5}\>\> Compute the strongly connected components of $G \setminus \lbrace w\rbrace$. \\
{\small 6}\>\> \textbf{for} each strongly connected component $C$ of $G \setminus \lbrace w\rbrace$ \textbf{do} \\
{\small 7}\>\>\> Recursively compute the $2$-vccs of $G[C\cup \lbrace w \rbrace]$ and \textbf{output} them.
\end{tabbing}
\end{meinAlgorithmus}
\end{figure}
\begin{Theorem}\label{def:LaufzeitAll2VsCCsRaedJaberi}
 Algorithm \ref{algo:simpleAlgorithmus2VsCCsRaedJaberi} runs in $O(nm)$ time.
\end{Theorem} 
\begin{proof} The strong articulation points of a directed graph can be computed in linear time using the algorithm of Italiano et al. \cite{lite2}. The strongly connected components of a directed graph can also be computed in linear time using Tarjan's algorithm \cite{lite5}.
Let $C_1,C_2,\ldots,C_t$ be the strongly connected components of  $G \setminus \lbrace w\rbrace$. Clearly, the edge sets of the subgraphs $G[C_1 \cup \lbrace w\rbrace],G[C_2 \cup \lbrace w\rbrace],\ldots,G[C_t \cup \lbrace w\rbrace]$ are disjoint. Thus the total cost at each recursion level is $O(m)$. Since the vertex set of a graph in a recursive call is smaller than the original vertex set, the recursion depth is at most $n$. Thus the total time is $O(nm)$.
\end{proof}
\section{The relationship between $2$-vertex-connected components and dominator trees} \label{def:rel2vccsdts}
In this section, we prove a connection between the $2$-vccs of a directed graph and dominator trees.
\begin{theorem}\label{def:threl2vccsanddt}
Let $G=(V,E)$ be a strongly connected graph and let $v$ be an arbitrary vertex in $G$. Let $C^{2vc}$ be a $2$-vcc of $G$. Then either all elements of $C^{2vc}$ are direct successors of some vertex $w\notin C^{2vc} $ or all elements $C^{2vc}\setminus \lbrace w\rbrace$ are direct successors of some vertex $w \in C^{2vc}$ in the dominator tree $DT(v)$ of the flowgraph $G(v)$.
\end{theorem}
\begin{proof} We consider two cases:
\begin{enumerate}[1.]
\item $v \in C^{2vc}$. In this case, all elements of $C^{2vc} \setminus \lbrace v\rbrace$ are direct successors of $v$ in $DT(v)$ since for every vertex $x$ of $C^{2vc} \setminus \lbrace v\rbrace$, there are two vertex-disjoint paths from $v$ to $x$ in $G[C^{2vc}]$, hence in $G(v)$.
\item $v \notin C^{2vc}$. Then there is a vertex $x \in C^{2vc}$ such that $imd(x)=w$ and $w \notin C^{2vc}$. We show, by contradiction, that $w$ is a dominator for every vertex of $C^{2vc}$. Assume that there is some vertex $y \in C^{2vc}$, $y\neq x$ such that $w \notin dom(y)$. Consequently, there exists a path $p$ from $v$ to $y$ not containing $w$. This path enters $C^{2vc}$ in vertex $u\in C^{2vc}$. This means that the vertices of $p$ from $v$ to $u$ are outside of $C^{2vc}$. Moreover, there are two vertex-disjoint paths from $u$ to $x$ in $G[C^{2vc}]$. Thus, there is a path from $v$ to $x$ not containing $w$. Therefore, we have $w \notin dom(x)$, which contradicts that $imd(x)=w$.
Now we consider two cases:
\begin{enumerate}[a)]
\item All paths from $w$ to $x$ are completely outside of $G[C^{2vc}]$. Then $x$ is a dominator of all vertices $y\in C^{2vc}$. (Assume that there is a path from $v$ to $y$ avoiding $x$. By the above, $w$ is on this path. We can extend $p$ inside $C^{2vc}$ to reach $x$, contradicting the assumption all paths from $w$ to $x$ are completely outside of $G[C^{2vc}]$.) \\ 
For every vertex $y\in C^{2vc}$, $x$ is the immediate dominator of $y$ in $DT(v)$, since there are two vertex-disjoint paths from $x$ to $y$. 
\item There are at least two vertex-disjoint paths $p_1,p_2$ from $w$ to $x$ such that $p_1$ enters $C^{2vc}$ in vertex $y$ and $p_2$ enters $C^{2vc}$ in vertex $y'$ with $y\neq y'$. Since there are a path from $y$ to $y'$ and a path from $y'$ to $y$ in $G[C^{2vc}]$, there are two vertex-disjoint paths from $w$ to $y$ and two vertex-disjoint paths from $w$ to $y'$ in $G(v)$. Therefore, the vertices $y,y'$ are direct successors of $w$ in $DT(v)$. Now we prove that every vertex $z\in C^{2vc}\setminus \lbrace x,y,y'\rbrace$ is also direct successor of $w$. There are two case to consider:
\begin{enumerate}[(i)]
\item All paths from $y$ to $z$ and all paths from $y'$ to $z$ have a vertex $z' \in C^{2vc}$ in common with $z'\notin \lbrace y,y',z \rbrace$. Consequently, all the paths from $y$ to $z$ contain $z'$ in $G[C^{2vc}]$, where $z'\notin \lbrace y,z \rbrace$. Hence, $z'$ is a strong articulation point in $G[C^{2vc}]$ by \cite[Lemma $2.1$]{lite2} of Italiano et al., which contradicts that $G[C^{2vc}]$ is $2$-vcc of $G$.
\item To interrupt all paths from $\lbrace y,y'\rbrace$ to $z$, one has to remove at least two vertices. We add a vertex $s\notin V$ and two edges $(s,y),(s,y')$ to $G$. Clearly, $s$ and $z$ are not adjacent. A separator of all paths from $s$ to $z$ is a set of vertices whose removal interrupts all paths from $s$ to $z$. A minimal separator of all paths from $s$ to $z$ has two vertices. By Menger's Theorem ($1927$) there are two vertex-disjoint paths from $s$ to $z$. Thus, there exist a path $p$ from $y$ to $z$ and a path $p'$ from $y'$ to $z$ in $G[C^{2vc}]$ such that $p,p'$
are vertex-disjoint. As a consequence, there are two vertex-disjoint paths from $w$ to $z$ in $G(v)$. Therefore, $z$ is direct successor of $w$ in $DT(v)$.
\end{enumerate}
\end{enumerate}
\end{enumerate}
\end{proof}
By $M(w)$ we denote the set of direct successors of vertex $w$ in the dominator tree of a flowgraph. Algorithm \ref{algo:All2VsCCsRaedJaberitherel2vccanddt} shows a new algorithm for computing all the $2$-vccs of a strongly connected graph $G$ using Theorem \ref{def:threl2vccsanddt}.
\begin{figure}[h]
\begin{meinAlgorithmus}[\textbf{All2VsCCs$(G)$}]\label{algo:All2VsCCsRaedJaberitherel2vccanddt}\rm\quad\\[-5ex]
\begin{tabbing}
\quad\quad\=\quad\=\quad\=\quad\=\quad\=\quad\=\quad\=\quad\=\quad\=\kill
\textbf{Input:} A strongly connected graph $G=(V,E)$.\\
\textbf{Output:} The $2$-vccs of $G$.\\
{\small 1}\> \textbf{if} $G$ is $2$-vertex-connected \textbf{then} \\
{\small 2}\>\> \textbf{Output} $V$. \\
{\small 3}\> \textbf{else}\\
{\small 4}\>\> Compute the strong articulation points of $G$.\\
{\small 5}\>\> Choose a vertex $v \in V$ that is not a strong articulation point of $G$.\\
{\small 6}\>\> Compute the dominator trees $DT(v)$ and $DT^{R}(v)$. \\
{\small 7}\>\> Choose a dominator tree of  $\lbrace DT(v),DT^{R}(v) \rbrace$ that contains more \\
{\small 8}\>\>\> non-trivial dominators. \\
{\small 9}\>\> \textbf{for} each vertex $w\in V$ \textbf{do}\\
{\small 10}\>\>\> \textbf{if} $|M(w)| \geq 2$ \textbf{then}  \\
{\small 11}\>\>\>\> \textbf{if} $G[M(w)\cup \lbrace w \rbrace]$ is not strongly connected \textbf{then}\\
{\small 12}\>\>\>\>\>  Compute the strongly connected components of $G[M(w)\cup \lbrace w\rbrace]$.\\
{\small 13}\>\>\>\>\>  \textbf{for} each strongly connected component $C$ of $G[M(w)\cup \lbrace w \rbrace]$ \textbf{do}\\
{\small 14}\>\>\>\>\>\> \textbf{if} $|C|\geq 3$ \textbf{then} \\
{\small 15}\>\>\>\>\>\>\> Recursively compute the $2$-vccs of $G[C]$ and \textbf{output} them.  \\
{\small 16}\>\>\>\> \textbf{else}\\
{\small 17}\>\>\>\>\> Recursively compute the $2$-vccs of $G[M(w)\cup \lbrace w \rbrace]$ and \textbf{output} them.  
\end{tabbing}
\end{meinAlgorithmus}
\end{figure}

Algorithm \ref{algo:All2VsCCsRaedJaberitherel2vccanddt} works as follows. First, line $1$ tests whether the strongly connected graph $G$ is $2$-vertex-connected using the algorithm of Italiano et al. \cite{lite2}, and if it is, line $2$ outputs $V$. Otherwise, the algorithm finds a dominator tree whose depth is at least $2$ as follows. Line $5$ chooses a vertex $v \in V$ which is not a strong articulation point of $G$. Then, line $6$ computes the dominator trees $DT(v)$ and $DT^{R}(v)$ since at least one of them has non-trivial dominators. In order to reduce the recursion depth, we  choose a dominator tree of $\lbrace DT(v),DT^{R}(v) \rbrace$ that contains more non-trivial dominators. $M(w)$ is the set of direct successors of vertex $w$ in the dominator tree that is chosen in line $7$. For each vertex $w\in V$ with $|M(w)| \geq 2$, the algorithm tests whether if $G[M(w)\cup \lbrace w \rbrace]$ is strongly connected, and if it is, line $17$ recursively computes the $2$-vccs of $G[M(w)\cup \lbrace w \rbrace]$. Otherwise, the for loop of lines $13$--$15$ recursively computes the $2$-vccs of $G[C]$ for every strongly connected component $C$ of $G[M(w)\cup \lbrace w \rbrace]$.
\begin{theorem}\label{def:runtimealgorelationsccdominatortree}
Algorithm \ref{algo:All2VsCCsRaedJaberitherel2vccanddt} runs in $O(nm)$ time.
\end{theorem}
\begin{proof} Let $v,w$ be distinct vertices in $G$. Then the edge sets of the subgraphs $G[M(v)\cup \lbrace v \rbrace],G[M(w)\cup \lbrace w \rbrace]$ are disjoint since these subgraphs have at most one vertex in common. The rest of the proof is similar to the proof of Theorem \ref{def:LaufzeitAll2VsCCsRaedJaberi}.
\end{proof}
\section{Computing $3$-vertex-connected components of a directed graph} \label{def:Psec7}
The \textit{$k$-vertex-connected components} of a directed graph are its maximal $k$-vertex-connected subgraphs. This definition is a natural generalization of $2$-vccs which are defined by Italiano et al. \cite{lite2}.
In this section, we present an algorithm for computing the $3$-vccs of a directed graph. Our algorithm is based on the following Lemma, which is the obvious generalization of Lemma \ref{def:RSAKszken2Vfken}.
\begin{Lemma}\label{def:RSAKszken3Vfken}
Let $G=(V,E)$ be a $2$-vertex-connected directed graph and let $X\subset V$ be a vertex-cut of $G$ with $|X|=2$. Let $C^{3vc}$ be a $3$-vcc of $G$. Then all vertices of $ C^{3vc}\setminus X$ lie in a strongly connected component $C$ of $G \setminus X$, i.e. $C^{3vc}\setminus X\subseteq C$.
\end{Lemma} 
\begin{proof} Because $C^{3vc}$ is a $3$-vcc of $G$, the directed graph $G[C^{3vc}]$ does not contain any vertex-cut $Y\subset V$ with $|Y|<3$ by definition. Hence, $G[C^{3vc}\setminus X]$ is strongly connected. Furthermore, $G[C^{3vc}\setminus X]$ is a subgraph of $G \setminus X$. Therefore, $C^{3vc}\setminus X$ is a subset of a strongly connected component $C$ of $G \setminus X$, i.e. $C^{3vc}\setminus X \subseteq C$. 
\end{proof}
\begin{Corollary}\label{def:CorollaryRSAKszken3Vfken}
Let $G=(V,E)$ be a $2$-vertex-connected directed graph and let $X\subset V$ be a vertex-cut of $G$ with $|X|=2$. The $3$-vccs of $G$ lie in the subgraphs $G[C_1 \cup X],G[C_2 \cup X],\ldots,G[C_t \cup X]$, where $C_1,C_2,\ldots,C_t$ are the strongly connected components of $G \setminus X$. 
\end{Corollary}
Algorithm \ref{algo:simpleAlgorithmus3VsCCsRaedJaberi} shows our algorithm for computing the $3$-vccs of a directed graph $G$. 
\begin{figure}[htbp]
\begin{meinAlgorithmus}[\textbf{3VCCs$(G)$}]\label{algo:simpleAlgorithmus3VsCCsRaedJaberi}\rm\quad\\[-5ex]
\begin{tabbing}
\quad\quad\=\quad\=\quad\=\quad\=\quad\=\quad\=\quad\=\quad\=\quad\=\kill
\textbf{Input:} A directed graph $G=(V,E)$.\\
\textbf{Output:} The $3$-vccs of $G$.\\
{\small 1}\> \textbf{if} $G$ is $3$-vertex-connected \textbf{then}  \\
{\small 2}\>\> \textbf{Output} $V$. \\
{\small 3}\> \textbf{else} \textbf{if} $G$ is $2$-vertex-connected \textbf{then} \\
{\small 4}\>\> Find a vertex-cut $X$ of $G$.\\
{\small 5}\>\> Compute the strongly connected components of $G \setminus X$. \\
{\small 6}\>\> \textbf{for} each strongly connected component $C$ of $G \setminus X$ \textbf{do} \\
{\small 7}\>\>\> Recursively compute the $3$-vccs of $G[C\cup X]$ and \textbf{output} them.\\
{\small 8}\> \textbf{else}\\
{\small 9}\>\> Compute the $2$-vccs of $G$.\\
{\small 10}\>\> \textbf{for} each $2$-vcc $C^{2vc}$ of $G$ \textbf{do} \\
{\small 11}\>\>\> Recursively compute the $3$-vccs of $G[C^{2vc}]$ and \textbf{output} them.
\end{tabbing}
\end{meinAlgorithmus}
\end{figure}
This algorithm works as follows.
First, line $1$ tests whether the directed graph $G$ is $3$-vertex-connected using Gabow's algorithm \cite{lite10}, and if it is, line $2$ outputs $V$. Otherwise, the algorithm tests whether $G$ is $2$-vertex-connected using the algorithm of Italiano et al. \cite{lite2}. If $G$ is $2$-vertex-connected, then
line $4$ finds a vertex-cut $X$ of $G$ using Gabow's algorithm \cite{lite10} and the for loop of lines $6$--$7$ recursively computes the $3$-vccs of $G[C\cup X]$ for each strongly connected component $C$ of $G \setminus X$. This is correct by Corollary \ref{def:CorollaryRSAKszken3Vfken}. If $G$ is neither $3$-vertex-connected nor $2$-vertex-connected, then line $9$ computes the $2$-vccs of $G$ using Algorithm \ref{algo:simpleAlgorithmus2VsCCsRaedJaberi} and the for loop of lines $10$--$11$ recursively computes the $3$-vccs of $G[C^{2vc}]$ for each $2$-vcc $C^{2vc}$ of $G$.
\begin{Theorem}\label{def:LaufzeitAll3VsCCsRaedJaberi}
Algorithm \ref{algo:simpleAlgorithmus3VsCCsRaedJaberi} runs in $O(n^{3}m)$ time.
\end{Theorem} 
\begin{proof} Let $G$ be a directed graph. The vertex connectivity $\kappa$ and a corresponding separator in $G$ can be found using Gabow's algorithm \cite{lite10} in $O((n+min \lbrace \kappa^{5/2}$,$\kappa n^{3/4}\rbrace)m)$ time. Furthermore, $2$-vertex-connectivity can be tested using the algorithm of Italiano et al. \cite{lite2} in linear time. Let $C_1,C_2,\ldots,C_t$ be the strongly connected components of $G \setminus X$ (see lines $5$--$7$). Since $|X|=2$, we have $|E[X]|<3$. The edge sets of the subgraphs $G[C_1 \cup X] \setminus E[X],G[C_2 \cup X] \setminus E[X],\ldots,G[C_t \cup X] \setminus E[X]$ are disjoint. By Theorem \ref{def:LaufzeitAll2VsCCsRaedJaberi}, the $2$-vccs $C_1^{2vc},C_2^{2vc},\ldots,C_l^{2vc}$ of $G$ can be computed in $O(nm)$ time. Moreover, by Fact \ref{def:2VCsAtMostOneCommonVertex}, the edge sets of the subgraphs $G[C_1^{2vc}],G[C_2^{2vc}],\ldots,G[C_l^{2vc}]$ are disjoint (see lines $10$--$11$).
Let $G_1,G_2,\ldots ,G_l$ be the subgraphs which are considered at any level of the recursion, let $n_i$ be the number of vertices of $G_i$ and let $m_i$ be the number of edges of $G_i$, where $1 \leq i \leq l$.
Then, the total cost at each recursion level is $n_1m_1+n_2.m_2+\ldots +n_lm_l \leq (n_1+n_2+\ldots+n_l)m \leq n^{2}m$ since $n_i \leq n$ and $l \leq n$.
Since the vertex set of a graph in a recursive call is smaller than the original vertex set, the recursion depth is at most $n$. Thus the total time is $O(n^{3}m)$.
\end{proof}
It is not difficult to see that any two $k$-vccs of a directed graph have at most $k-1$ vertices in common. We can compute the $k$-vccs of a directed graph using the following Lemma, which is the obvious generalization of Lemma \ref{def:RSAKszken3Vfken}.
\begin{Lemma}\label{def:RSAKszkenkVfken}
Let $G=(V,E)$ be a $(k-1)$-vertex-connected directed graph and let $X\subset V$ be a vertex-cut of $G$ with $|X|=k-1$. Let $C^{kvc}$ be a $k$-vcc of $G$. Then all vertices of $ C^{kvc}\setminus X$ lie in a strongly connected component $C$ of $G \setminus X$, i.e. $C^{kvc}\setminus X\subseteq C$.
\end{Lemma}
\emph{Proof}: The proof is similar to the proof of Lemma \ref{def:RSAKszken3Vfken}. \qed
\begin{Theorem}\label{def:calkvccsrunningtime}
The $k$-vccs of a directed graph can be computed in $O(mn^{2k-3})$ time.
\end{Theorem}
\begin{proof} We can modify Algorithm \ref{algo:simpleAlgorithmus3VsCCsRaedJaberi} by replacing \textquotedblleft $2$-vertex-connected\textquotedblright\ with \textquotedblleft $(k-1)$-vertex-connected\textquotedblright\ and by replacing \textquotedblleft $3$-vertex-connected\textquotedblright\ with \textquotedblleft $k$-vertex-connected\textquotedblright. The modified algorithm can compute the $k$-vccs of a directed graph. Its running time is bounded by the product of the recursion depth, $n$ times the cost for computing $(k-1)$-vccs and vertex-cut. One can easily prove by induction on $k$ that the running time of the modified algorithm is $O(mn^{2k-3})$.
\end{proof}

\section{Applications} \label{def:Psec8}
In this section, we consider three applications of the new algorithms.
 \begin{Problem}\label{def:2vcspcsproblem}
Given a directed graph $G=(V,E)$, find a minimum cardinality set $E^{*}\subseteq E$ such that the $2$-vccs of $G$ coincide with the $2$-vccs of the graph $G^{*}=(V,E^{*})$.  
\end{Problem} 
Clearly, the smallest $2$-vertex-connected spanning subgraph of a $2$-vertex-connected directed graph is a special case of problem \ref{def:2vcspcsproblem} when $G$ is $2$-vertex-connected. Therefore, by the results from \cite{lite14,lite8} problem \ref{def:2vcspcsproblem} is NP-hard.  
\begin{Lemma} \label{def:approxalgorfor2vcspcsproblem}
There is a $1.5$ approximation algorithm for problem \ref{def:2vcspcsproblem} with running time $O(nm)$.
\end{Lemma}
\begin{proof} First, we compute all the $2$-vccs $ C^{2vc}_1,C^{2vc}_2,\ldots, C^{2vc}_t$ of the directed graph $G$ using Algorithm \ref{algo:simpleAlgorithmus2VsCCsRaedJaberi}. The edges of the set $E\setminus (E[C^{2vc}_1]\cup E[C^{2vc}_2]\cup \ldots \cup E[C^{2vc}_t])$ are irrelevant.
Let $E_{opt}$ an optimal solution for problem \ref{def:2vcspcsproblem}. Then, by Fact \ref{def:2VCsAtMostOneCommonVertex}, we have $E_{opt}=E_1\cup E_2\cup\ldots \cup E_t$, where  
$E_i$ an optimal solution for the subgraph $G[C^{2vc}_i]$. Let $opt=|E_{opt}|$ and $opt_i=|E_i|$. Then, $opt=\sum_{1\leq i\leq t } opt_i$. In $2000$, Cheriyan und Thurimella \cite{lite7} gave a $(1+1/k)$-approximation algorithm for the problem of finding a minimum-size $k$-vertex-connected spanning subgraph of a directed graph with $m$ edges. This algorithm runs in $O(km^{2})$ time. In $2011$, Georgiadis \cite{lite8} improved the running time of the algorithm of Cheriyan und Thurimella from $O(m^{2})$ to $O(m\sqrt{n}+n^{2})$ for $k=2$. 
This improved algorithm \cite{lite8} preserves the $1.5$ approximation guarantee of the Cheriyan-Thurimella algorithm for $k=2$. Let $E_i^{'}$ be an edge set obtained by running the improved algorithm \cite{lite8} on the subgraph $G[C^{2vc}_i]$. Then, we have $\sum_{1\leq i\leq t } |E_i^{'}|\leq 1.5 \sum_{1\leq i\leq t }  opt_i \leq 1.5 opt$ because the edge sets of $G[C^{2vc}_i],1\leq i \leq t,$ are disjoint. The total running time is $O(\sum_{1\leq i\leq t }(|E[C^{2vc}_i]|\sqrt{|C^{2vc}_i|}+|C^{2vc}_i|^{2}) +nm)=O(nm)$ because $\sum_{1\leq i\leq t }(|E[C^{2vc}_i]|\sqrt{|C^{2vc}_i|}+|C^{2vc}_i|^{2})\leq \sqrt{n} \sum_{1\leq i\leq t} |E[C^{2vc}_i]| +\sum_{1\leq i\leq t} |C^{2vc}_i|^{2} \leq m\sqrt{n}+\sum_{1\leq i\leq t} |C^{2vc}_i|^{2} $ and $O(\sum_{1\leq i\leq t} |C^{2vc}_i|^{2})=O(n^{2})$ by Lemma \ref{def:thenumberofverticesofvccs}. 
\end{proof}
\begin{Lemma}\label{def:thenumberofverticesofvccs}
Let $G=(V,E)$ be a directed graph and let $ C^{2vc}_1,C^{2vc}_2,\ldots, C^{2vc}_t$ be the $2$-vccs of $G$.
Then $\sum_{1\leq i\leq t} |C^{2vc}_i|<3n$.
\end{Lemma}
\begin{proof} We construct a new graph $G_c=(V_c,E_c)$ from $G$ called a component graph as follows. For each $2$-vcc $C^{2vc}_i$ of $G$, we add a vertex $v_i$ to $V_c$. Let $C^{2vc}_i,C^{2vc}_j$ be distinct $2$-vccs of $G$. If $C^{2vc}_i,C^{2vc}_j$ have a vertex $w$ in common, then we add a vertex $w_{*}$ to $V_c$ and two undirected edges $\lbrace v_i,w_{*}\rbrace, \lbrace w_{*},v_j\rbrace$ to $E_c$. Since $G_c$ is a tree or a forest, $\sum_{1\leq i\leq t} |C^{2vc}_i|\leq |V|+|E_c|\leq n+n+t-1<3n$.
\end{proof}
 \begin{Problem}\label{def:sz2vcspcsproblem}
Given a strongly connected graph $G=(V,E)$, find a minimum cardinality set $E^{*}\subseteq E$ such that the $2$-vccs of $G$ coincide with the $2$-vccs of the directed graph $G^{*}=(V,E^{*})$ and $G^{*}$ is strongly connected.
\end{Problem} 
\begin{Lemma} \label{def:approxalgorforapproxalgorfor2vcspcsproblem}
There is an $5/3$ approximation algorithm for problem \ref{def:sz2vcspcsproblem} with running time $O(nm)$.
\end{Lemma}
\begin{proof} If we contract the $2$-vccs of $G$ that overlap into a super vertex, then we obtain a directed graph, which we call the coarsened graph of $G$. The edge sets within the $2$-vccs of $G$ and the edge set between $2$-vccs of $G$ are disjoint. We split the approximation problem \ref{def:sz2vcspcsproblem} into two independent porblems: problem \ref{def:2vcspcsproblem} and the minimum strongly-connected spanning subgraph problem. In $2003$, Zhao et al. \cite{lite9} gave a linear time $5/3$-approximation algorithm for the minimum strongly-connected spanning subgraph problem. We run this algorithm on the coarsened graph of $G$.\end{proof}
 \begin{Problem}\label{def:2vcs2vcvgproblem}
 Given a directed graph $G=(V,E)$, find a minimum cardinality set $E^{*}\subseteq E$ such that the $2$-vccs of $G$ coincide with the $2$-vccs of $G^{*}=(V,E^{*})$ and 
 the $2$-vccs of the coarsened graph of $G$ coincide with the $2$-vccs of the coarsened graph of $G^{*}$.
\end{Problem}
\begin{Lemma} \label{def:approxalgorfor2vcs2vcvgproblem1}
There is an $1.5$ approximation algorithm for problem \ref{def:2vcs2vcvgproblem} with running time $O(nm)$.
\end{Lemma}
\begin{proof} The proof is similar to the proof of Lemma \ref{def:approxalgorfor2vcspcsproblem} (we apply the same method on the graph $G$ and on the coarsened graph of $G$). 
\end{proof}
\section{Open problems}
We leave as an open problem whether the $2$-vccs of a directed graph that contain a certain vertex can be computed in linear time. Another open problem is whether the computing of $3$-vccs of a directed graph can be done in $O(nm)$ time.
\addcontentsline{toc}{section}{Acknowledgements}
\section*{Acknowledgements.}
The Author would like to thank Martin Dietzfelbinger for helpful comments and interesting discussions.

\end{document}